\documentclass[letterpaper,journal,twocolumn]{IEEEtran}
\usepackage[mathcal]{euscript}

\usepackage{graphicx}
\usepackage{cite}
\usepackage{verbatim}
\usepackage{color}
\usepackage{amssymb,amsfonts,amsmath,amsthm,amscd}
\usepackage{mathrsfs}
\usepackage{bigints}

\newtheorem{theorem}{Theorem}
\newtheorem{proposition}{Proposition}
\newtheorem{lemma}{Lemma}

\newtheorem{remark}{Remark}

\newtheorem{definition}{Definition}

\def\sU{{\mathsf U}}
\def\sX{{\mathsf X}}
\def\sY{{\mathsf Y}}

\def\cB{{\mathcal B}}

\def\cF{{\mathcal F}}

\def\cM{{\mathcal M}}
\def\cP{{\mathcal P}}

\def\law{\operatorname{Law}}

\def\d{\operatorname{d}\!}
\def\E{{\mathbf E}}
\def\Pr{{\mathbf{P}}}

\def\1{{\mathsf 1}}
\def\deq{:=}

\def\R{{\mathbb R}}
\def\d{\mathrm{d}}

\def\eps{\varepsilon}
\def\d{{\mathrm d}}
\def\E{{\mathbb E}}

\def\1{{\mathbf 1}}
\def\eps{\varepsilon}
\def\cA{{\cal A}}
\def\cB{{\cal B}}

\def\cF{{\cal F}}

\def\cM{{\cal M}}

\def\cP{{\cal P}}

\def\sY{{\cal Z}}

\def\sP{{\mathsf P}}

\def\sU{{\mathsf U}}

\def\sX{{\mathsf X}}
\def\sY{{\mathsf Y}}

\def\Reals{{\mathbb R}}

\def\tpi{\tilde{\pi}}

\title{Sequential Empirical Coordination\\
Under an Output Entropy Constraint }
\author{Ehsan Shafieepoorfard and Maxim Raginsky,~\IEEEmembership{Senior Member,~IEEE}%
\thanks{E.~Shafieepoorfard and M.~Raginsky are with the Department of Electrical and Computer Engineering and the Coordinated Science Laboratory, University of Illinois at Urbana--Champaign, Urbana, IL 61801, USA; {\tt \{shafiee1,maxim\}@illinois.edu}.}
\thanks{Research supported in part by the NSF under CAREER award no. CCF-
1254041, and by the Center for Science of Information (CSoI), an NSF
Science and Technology Center, under grant agreement CCF-0939370. Prelminary version of this work was presented at the IEEE Conference on Decision and Control, Las Vegas, NV, December 2016.}}

\begin{document}

\maketitle

\thispagestyle{empty}




\begin{abstract}
This paper considers the problem of sequential
empirical coordination, where the objective is to achieve a
given value of the expected uniform deviation between state-action
empirical averages and statistical expectations under a
given strategic probability measure, with respect to a given universal
Glivenko-Cantelli class of test functions. A communication
constraint is imposed on the Shannon entropy of the resulting
action sequence. It is shown that the fundamental limit on
the output entropy is given by the minimum of the mutual
information between the state and the action processes under
all strategic measures that have the same marginal state process
as the target measure and approximate the target measure
to desired accuracy with respect to the underlying Glivenko--Cantelli
seminorm. The fundamental limit is shown to be
asymptotically achievable by tree-structured codes.
\\ \\
\end{abstract}




\section{Introduction}
\label{ssec:introduction}

\IEEEPARstart{D}{ecision-making} in the presence of limited communication or information acquisition resources has long been a major topic of interest in the study of both networked control systems \cite{witsenhausen1979structure,teneketzis1980communication,walrand1983optimal,mitter2001control,shafieepoorfard2016rationally} and economics \cite{sims2003implications,luo2006rational,huang2007rational, van2010information,Tutino2012, matvejka2015rationally, martin2016strategic}. The problem becomes more intricate when causality constraints are imposed, and decisions must be made in real time. There are many recent works on this topic dealing with various types of information constraints and  structural assumptions about the source (see, e.g., \cite{borkar2001optimal, borkar2005sequential, teneketzis2006structure, yuksel2008optimal, charalambous2014nonanticipative, linder2014optimal, tanaka2014semidefinite, stavrou2017zero}).

Reconstruction of an information source from its compressed version subject to a fidelity criterion is the focus of rate-distortion theory \cite{berger1971rate}; there is also a sequential generalization of rate-distortion theory \cite[Ch.~5]{TatikondaThesis} to reconstruction problems with causality constraints {(additional relevant works, motivated by video and perceptual coding, include \cite{vis2000seq,maishwar2011seq,yang2011video})}. Moreover, when the reconstruction is fed into a controller that can act on the information source, and if one can establish some form of the separation principle between estimation and control, the methods of sequential rate-distortion theory can be brought to bear on the problem of optimal quantizer design for this problem of \textit{control under communication constraints} \cite{witsenhausen1971separation,varaiya1983causal,bansal1989simultaneous, bacsar1994optimum, mitter2001control, borkar2001markov, tatikonda1999control, matveev2004problem, yuksel2013stochastic, tanaka2015lqg}. 

However, there is an alternative perspective on the problem of compressed representations in networked control systems -- that of \textit{empirical coordination} under communication constraints. The problem of coordination, first introduced in the information theory literature by Cuff et al.~\cite{cuff2010coordination} (see also \cite{yassaee2015channel}), can be stated as follows: Consider a finite collection of decision makers (or DM's, for short), who wish to generate actions in response to a random state variable  according to some prescribed policy, but can only receive information about the state  over finite-capacity noiseless digital links. Suppose that we have a large number of independent and identically distributed (i.i.d.) copies of the state, and let the DM's generate a sequence of actions based on the information they receive about this state sequence. What are the minimal communication requirements (in bits per copy), to guarantee that the long-term empirical frequencies of realized states and actions approximate, to desired accuracy, the ideal joint probability law of states and actions induced by the marginal law of the state and the policy?

Cuff et al.~\cite{cuff2010coordination} assume that both the state and the actions take values in finite sets, and measure the quality of approximation by the total variation distance between the empirical distribution of states and actions and the target joint distribution. However, this criterion is inapplicable to continuous-valued states and/or actions with nonatomic probability laws because the total variation distance between any nonatomic probability measure and any discrete probability measure attains its maximal value. To resolve this issue, Raginsky \cite{raginsky2013empirical} proposed a relaxed approximation criterion: Fix a suitable class of bounded real-valued test functions on the space of all state-action pairs and consider the worst-case deviation between their empirical averages and their expectations with respect to the target measure. Under the regularity assumption that the class of test functions has the so-called \textit{universal Glivenko--Cantelli property} (cf.~\cite{van2013universal} and references therein, as well as Section~\ref{ssec:notation} for definitions), Raginsky \cite{raginsky2013empirical} obtained a full information-theoretic characterization of the minimal communication requirements for empirical coordination. Since any uniformly bounded class of real-valued functions on a finite set is universal Glivenko--Cantelli, the framework of \cite{cuff2010coordination} emerges as a special case.

In this paper, we present an extension of the empirical coordination framework of \cite{raginsky2013empirical} to the sequential setting: We consider a two-terminal network consisting of a sender and a decision-maker (DM). The sender observes $N$ independent copies of a discrete-time state process of fixed finite duration $T$. {It is useful to think of each copy as input data for a task, which involves taking $T$ actions contingent on the states. Completion of the task involves implementing a fixed causal policy on the state process corresponding to that task.} However, the DM has no direct access to the state processes. Instead, the sender can communicate with the DM over a finite-capacity noiseless digital channel, {and the idea is to exploit statistical regularity across tasks to reduce the amount of communication needed to guarantee that, on average, the performance of the DM on all the tasks resembles the ideal joint distribution of states and actions prescribed by the policy.} Thus, we are interested in the {\textit{communication complexity} of coordination, i.e.,} the minimal amount of communication needed to guarantee that, in the limit as $N \to \infty$, the empirical distribution of states and actions at each time $t \in [T]$ can approximate the state-action distribution induced by the state process law and by the policy specification. The coding scheme employed by the sender must satisfy the sequentiality constraint: The signal transmitted by the sender to the DM at time $t$ may only depend on the realizations of the state processes up to time $t$. Following Tatikonda \cite{TatikondaThesis}, we quantify the communication resources by the Shannon entropy of the signal process. Entropy constraints on the quantizer output are commonly used in causal source coding problems \cite{neuhoff1982causal}, where the compressed representation of the source at time $t$ may depend on the present and on the past source samples, but not on the future ones.

Our main contribution is a full information-theoretic characterization of the fundamental limit on the amount of communication from the sender to the DM in the setting of sequential empirical coordination. We refer to this fundamental limit as the \textit{sequential rate-distortion function for empirical coordination}. Specifically, we show that, for all large enough $N$, this fundamental limit can be achieved by means of tree-structured codes of the kind employed by Tatikonda \cite{TatikondaThesis}, and that no sequential scheme for empirical coordination can beat this fundamental limit.\footnote{The reference policy for generating actions contingent on the states may be randomized. However, we restrict the sequential encoder used by the sender and the sequential decoder used by the DM to be \textit{deterministic}. The reason for this is that, when randomized strategies are used in the absence of a noiseless feedback channel from the DM to the sender, the sender has to form beliefs about the actions taken by the DM, who will in turn form beliefs about the beliefs by the sender about the actions taken by the DM, and so on, leading to the so-called infinite regress of expectations (see, e.g.,\cite{sargent1991equilibrium}). This lack of precise knowledge on the part of the sender will accumulate over time. Restricting to deterministic strategies removes this problem: at any time $t$, the sender is strictly better informed than the receiver and can perfectly reconstruct the actions taken by the receiver.} While we do not make any structural assumptions on the state process (e.g., it is not assumed to be memoryless, Markov, ergodic, etc.), we assume that the target policy is feedforward (i.e., there is no functional dependence of future states on current and past actions).

Like other set-ups that include the interplay between information acquisition and decision-making, the sequential coordination problem considered here can also be interpreted in the framework of economics. It arises when a finite number of economic agents (or sectors) with constrained cognitive (or communication) resources \cite{sims2003implications} are subject to idiosyncratic economic shocks. A better-informed information sender -- such as a central bank or monetary authorities -- wishes to recommend optimal actions to all the agents through a common public signal. On average, though, the sender's optimal signaling strategy must take into account the limits on information-processing capacities of all of the agents. Our paper addresses several features of this set-up as well; however, we do not consider situations involving strategic motives, in which different players involved in the information exchange have biased or opposing objectives. Strategic considerations have been addressed recently, both in economics \cite{crawford1982strategic, rayo2010optimal, gentzkow2011bayesian} and in  information theory \cite{akyol2015strategic}.

\subsection{Contents of the paper}

The organization of the paper is as follows. We introduce the notation and basic concepts (in particular, Glivenko--Cantelli classes) in Section~\ref{ssec:notation}. The precise formulation of the sequential empirical coordination problem is given in Section~\ref{ssec:formulate}. The main results are presented in Section~\ref{ssec:mmain}, with some examples discussed in Section~\ref{sec:implicc}. Appendix~\ref{app:uGCappendix} contains a discussion of typicality in standard Borel spaces based on universal Glivenko--Cantelli classes. Two technical lemmas needed in the achievability proof are given in  Appendix~\ref{app:lemmas}.



\section{Preliminaries and Notations}
\label{ssec:notation}

All spaces in this paper are assumed to be {\itshape standard Borel spaces}, as defined below (for detailed treatments, see the lecture notes of Preston \cite{preston2008some} or Chapter 4 of Gray \cite{gray2011entropy}). 

\begin{definition} A measurable space $(\sX , \cB(\sX))$ is  {\itshape standard Borel} if it can be metrized with a metric $d$, such that: 1) $(\sX , d)$ is a complete separable metric space, and 2) $\cB(\sX)$ is the Borel $\sigma$-algebra, i.e., the smallest $\sigma$-algebra containing all open sets in $(\sX,d)$.
\end{definition}


\noindent We denote by $\cP(\sX)$ the space of all Borel probability measures on $\sX$, and by $M(\sX)$ the space of all bounded measurable functions $\sX \to \Reals$ equipped with the sup norm
\[
||f||_{\infty} \deq \sup_{x \in \sX} |f(x)|.
\]
We use the standard inner-product notation for integrals: given any signed Borel measure $\nu$ on $\sX$ and $f \in M(\sX)$,
\[
	\langle \nu,f \rangle \deq \int_\sX f(x) \,\nu(\d x) . 
\]
When $\nu \in \cP(\sX)$, we will also use the standard expectation notation $ \E_\nu \left[ f(X) \right]$. A \textit{Markov} (or \textit{stochastic}) \textit{kernel} with input space $\sX$ and output space $\sY$ is a mapping $K(\cdot | \cdot) : \cB(\sY) \times \sX \rightarrow [0,1]$, such that $K(\cdot|x) \in \cP(\sY)$ for every $x \in \sX$ and $K(B|\cdot) \in M(\sX)$ for every $B \in \cB(\sY)$. We denote the space of all such kernels by $\cM(\sY|\sX)$. Any $K \in \cM(\sY|\sX)$ maps $\cP(\sX)$ into $\cP(\sY)$:
\begin{align*}
\mu K(\cdot) \deq \int_{\sX} K(\cdot|x)\mu(\d x).
\end{align*}
Given a probability measure $\mu \in \cP(\sX)$ and a Markov kernel $K \in \cM(\sY|\sX)$, we denote by $\mu \otimes K$ the probability measure  on the product space $(\sX \times \sY, \cB(\sX) \otimes \cB(\sY))$ uniquely specified by its values on the rectangles $A \times B$, $A \in \cB(\sX), B \in \cB(\sY)$:
\[
	(\mu \otimes K)(A \times B) \deq \int_A K(B|x)\mu(\d x).
\]
If we let $A = \sX$ in the above definition, then we end up with with $\mu K(B)$. Note that product measures $\mu \otimes \nu$, where $\nu \in \cP(\sY)$, arise as a special case of this construction, since any $\nu \in \cP(\sY)$ can be realized as a Markov kernel $(B,x) \mapsto \nu(B)$. Conversely, given a random element $(X,Y)$ of $\sX \times \sY$, its probability law $\nu \in \cP(\sX \times \sY)$ can be disintegrated as $\mu \otimes K$, where $\mu(\cdot) = \nu(\cdot \times \sY) \in \cP(\sX)$ is the marginal distribution of $X$ and $K \in \cM(\sY|\sX)$ is a version of the conditional distribution of $Y$ given $X$.

\subsection{Universal Glivenko--Cantelli classes}

The notion of a {\itshape universal Glivenko-Cantelli class} \cite{van2013universal} (or uGC class for short) plays a central role in this paper. The main reason for adopting this notion is that it leads to a fruitful extension of the notion of typical sequences in standard Borel spaces \cite{raginsky2013empirical} (cf.~Appendix~\ref{app:uGCappendix} for a discussion). Here, we set up the notation and the definitions that will be needed in the sequel.

Given a class of measurable functions $\mathcal{F} \subseteq \{f \in M(\sX) : \|f\|_\infty \le 1\}$ and a signed Borel measure $\nu$ on $\sX$, we define the seminorm
\[
\Vert\nu\Vert_{\mathcal{F}}\deq\sup_{f\in\mathcal{F}}\left\vert
\langle\nu,f\rangle\right\vert .
\]
Let $[N] \deq \{1,\ldots,N\}$, for $N \in {\mathbb N}$, and let $x_{[N]} = (x_1,\ldots,x_N)$ denote an $N$-tuple of elements of $\sX$. The \textit{empirical measure} of $x_{[N]}$ is an element of $\cP(\sX)$, defined as\footnote{Since $\sX$ is a standard Borel space, all singletons are measurable and belong to $\cB(\sX)$.}
\[
\sP_{x_{[N]}}(\cdot) =\frac{1}{N}\sum_{n \in [N]}\delta_{x_n}(\cdot),
\]
where $\delta_{x}$ is the Dirac measure centered at $x$.
 \begin{definition} \label{def:uGC} A function class $\mathcal{F} \subset \{ f \in M(\sX) : \|f\|_\infty \le 1\}$ is a \textit{universal Glivenko--Cantelli class} (or a uGC class, for short) if
 \[
\left\Vert \sP_{X_{[N]}}-\mu\right\Vert _{\mathcal{F}}%
\xrightarrow{N\to\infty}0,\qquad\mu\text{-a.s.}
\]
for any $\mu \in \cP(\sX)$, where $X_1,X_2,\ldots$ is a stationary memoryless random process with marginal distribution $\mu$.
 \end{definition} 
For example, if $\sX = \R$, then the class 
\[
\cF \deq \left\{ f_z = \1_{(- \infty ,  z]} : z \in \R \right\}
\]
of indicator functions of semi-infinite intervals is a uGC class (this is the well-known Glivenko--Cantelli theorem, which explains the origin of the name ``universal Glivenko--Cantelli'').

\subsection{Information-theoretic preliminaries}

Throughout the paper, we rely on standard definitions and notions from information theory. The {\em relative entropy} (or {\em information divergence})\cite{PinskerBook} between $\mu,\nu \in \cP(\sX)$ is
\begin{align*}
	D(\mu \| \nu) &\deq \begin{cases}
	\left\langle \mu, \log \displaystyle\frac{\d\mu}{\d\nu} \right\rangle, & \text{ if $\mu \prec \nu$}  \\
	+\infty, & \text{otherwise},
\end{cases}
\end{align*}
where $\prec$ denotes absolute continuity of $\mu$ w.r.t. $\nu$, and $\d\mu/\d\nu$ is the Radon--Nikodym derivative. It is always nonnegative, and is equal to zero if and only if $\mu \equiv \nu$. The {\em Shannon mutual information} \cite{PinskerBook} in $(\mu,K) \in \cP(\sX) \times \cM(\sY|\sX)$ is defined as
\begin{equation}
	I(\mu,K) \deq D(\mu \otimes K \| \mu \otimes \mu K). 
\label{e:ShannonI}
\end{equation}
If $(X,Y)$ is a pair of random objects with $\law(X,Y) = \mu \otimes K$, then we will also write $I_{\mu,  K}(X; Y)$  to denote \eqref{e:ShannonI}. We will use standard identities for the mutual information and for the conditional mutual information, as can be found in \cite{cover2012elements}. We work with natural logarithms throughout the paper, so all entropies and mutual information is measured in \textit{nats}. 


\section{Problem formulation}

\label{ssec:formulate}

We now provide the precise formulation of the problem of sequential empirical coordination informally stated in Section~\ref{ssec:introduction}. The objective is for the sender to use minimal communication resources, so that the joint empirical distributions of the states observed by the sender and the actions generated by the receiver can mimic a given process law subject to a fidelity criterion. Since this problem involves causality considerations, we need to introduce the definition of a \textit{directed stochastic kernel} (see \cite{TatikondaThesis,tatikondamitter2009feedback} for a detailed presentation in the context of control and feedback information theory):

\begin{definition} \label{def:causcaus} Let $\sY_1,\ldots,\sY_M$ be a collection of Borel spaces. For any any $I \subseteq [M]$, let $\sY_I \deq \prod_{i \in I}\sY_i$. Fix any set $I = \{i_1,\ldots,i_K\}$ with $i_1 < i_2 < \ldots < i_K$, and let $I^c$ denote the complementary set $[M] \setminus I$. A \textit{directed stochastic kernel} between $\sY_{I^c}$ and $\sY_I$ is an element of $\cM(\sY_I|\sY_{[i_K-1]})$ that has the form
\[
K(\d y_I | y_{[i_K-1]}) = \bigotimes^K_{k=1} T_k(\d y_{i_k}|y_{[i_k-1]})
\]	
for a given collection of Markov kernels $T_k \in \cM(\sY_{i_k}|\sY_{[i_k-1]})$, $k = 1,\ldots,K$. We will denote the space of all such kernels by $\overrightarrow{\cM}(\sY_I|\sY_{I^c})$.
\end{definition}
\begin{remark}{\em {It is important to keep in mind that $\sY_{I^c}$ is not the input space of an element of $\overrightarrow{\cM}(\sY_I|\sY_{I^c})$. Rather, this notation is meant to distinguish the control variables from the observation variables, as explained in detail below.}}
\end{remark}
\noindent This definition naturally incorporates causality constraints. If we think of the index $i \in [M]$ as time and let the times $i_k \in I$ denote the instants when an action must be taken, then the Markov kernel $T_k$ prescribes the stochastic law for taking a random action at time $i_k$ on the basis of the `past' data $y_1,y_2,\ldots,y_{i_k-1}$. The stochastic kernel $K$ describes the overall sequential process of taking actions. A canonical way in which such stochastic kernels arise is to start with a random element $Y_{[M]}=(Y_1,\ldots,Y_M)$ of $\sY_{[M]}$ with probability law $\Pr$ and, for each $k \in [K]$, take $T_k$ to be the regular conditional probability distribution $\Pr_{Y_{i_k}|Y_{[i_k-1]}}$. This construction yields the  directed stochastic kernel
\[
 K_ (\d y_{I} | y_{[i_K-1]} ) \deq \displaystyle\bigotimes_{ k =1}^{K} \Pr_{Y_{i_k} | Y_{[i_k - 1]}} (\d y_{i_k} | y_{[i_k - 1]}).
\]

For the problem of empirical coordination, fix the state space $\sX$, the action space $\sU$, and the time horizon $T$. For each $t \in [T]$, introduce the copies $\sX_t$ and $\sU_t$ of $\sX$ and $\sU$, respectively. Let $\mu \in \cP(\sX_{[T]})$ denote the probability law of the state process $X_{[T]}=(X_1,\ldots,X_T)$, which can be disintegrated as the product of $T$ factors $\mu^{(t)} \in \cM(\sX_t | \sX_{[t-1]} )$:
$$
\mu(\d x_{[T]}) = \bigotimes^T_{t=1} \mu^{(t)}(\d x_t|x_{[t-1]}).
$$
Furthermore, let $\pi \in \overrightarrow{\cM}(\sU_{[T]} |\sX_{[T]})$ denote the directed stochastic kernel whose factors $\pi^{(t)} \in \cM(\sU_t| \sX_{[t]} \times \sU_{[t-1]})$  prescribe the causal policy, according to which the DM takes target actions in $\sU$ based on the past history of states and actions. The resulting joint probability law of states and actions, the so-called {\itshape strategic measure} $\mathbf{P}_\mu^\pi \in \cP(\sX_{[T]} \times \sU_{[T]})$, is given by 
\begin{equation}
\begin{split}
& \mathbf{P}_\mu^\pi (\d x_{[T]} , \d u_{[T]}) \\
& \deq \mu(\d x_{[T]}) \otimes \displaystyle\bigotimes_{t \in [T]} \pi^{(t)} (\d u_t | x_{[t]} , u_{[t-1]} ) \nonumber \\
&= \bigotimes_{t \in [T]} \mu^{(t)} (\d x_t | x_{[t-1]}) \otimes  \pi^{(t)} (\d u_t | x_{[t]} , u_{[t-1]} ) .
\end{split}
\end{equation}
The marginal law of $X_t$ under $\mathbf{P}_\mu^\pi$ will be denoted by $\mu_t \in \cP(\sX)$, and the marginal conditional law of $U_t$ given $X_t$ by $\pi_t \in \cM(\sU | \sX)$. 

Let $X_{[T],[N]} \deq \{X_{t,n} : t \in [T], n \in [N]\}$ be a $T \times N$ array of random elements of $\sX$, where $t\in [T]$ denotes the time index, while $n \in [N]$ enumerates the copies of $\mu$. For all $A \subset [T]$ and $B \subset [N]$, we will denote by $X_{A,B}$ the sub-array $(X_{t,n} : t \in A ; n \in B)$. For each $n \in [N]$ the columns $X_{[T],n} = (X_{t,n})_{t \in [T]}$ are i.i.d.\ copies of the state process with law $\mu$.  Similarly, let $\{ U_{t,n} : t \in [T] ; n \in [N] \}$, denoted by $U_{[T],[N]}$, be an array of random elements of $\sU$ such that for $n \in [N]$ the pair process $(X_{[T], n} , U_{[T], n})$ are i.i.d.\ copies of the state-action process with law $\mathbf{P}_\mu^\pi$. Notice that both arrays $X_{[T],[N]}$ and $U_{[T],[N]}$ are independent across $[N]$ while correlated across $[T]$. The empirical distribution of state-action pairs at time $t$ is given by
\[
\sP_{X_{t, [N]} , U_{t, [N]}} \deq \frac{1}{N}\sum_{n \in [N]}\delta_{X_{t,n} , U_{t,n}}.
\] 
Then, for any uGC class $\cF$, 
\begin{align}\label{eq:t_consistency}
\left\| \sP_{X_{t,[N]},U_{t,[N]}} - \mu_t \otimes \pi_t \right\|_\cF \xrightarrow{\text{\,\,a.s.\,\,}} 0 \qquad \text{as $N \to \infty$}.
\end{align}
Since \eqref{eq:t_consistency} holds for every $t \in [T]$, we have
\[
\frac{1}{T}\sum^T_{t=1} \left\| \sP_{X_{t,[N]},U_{t,[N]}} - \mu_t \otimes \pi_t \right\|_\cF \xrightarrow{\text{\,\,a.s.\,\,}} 0 \qquad \text{as $N \to \infty$}.
\]
That is, the realized empirical distributions of states and actions will be asymptotically consistent with the strategic measure $\mathbf{P}^\pi_\mu$, uniformly over $\cF$.

We now consider the following sequential coding problem involving an information sender (IS) and a decision-maker (DM). The IS can transmit messages to the DM over a finite-capacity channel. At each time $t$, the IS observes the state realizations $X_{[t] , [N]}$ and sends a message to the DM who will use this message and all previously received messages to generate the new $N$-tuple of actions $U_{t,[N]}$ using a deterministic policy. The goal is to ensure that the realized empirical distributions of states and actions approximate the strategic measure ${\mathbf P}^\pi_\mu$ to a given accuracy, while minimizing the communication resources. We will assess the quality of approximation using a fixed uGC class of test functions, while the communication resources will be measured in terms of the overall Shannon entropy of the messages sent by the IS to the DM.

\begin{definition}A {\itshape sequential N-code} is a collection $\gamma = (\gamma_t)_{t \in [T]}$ of measurable mappings
\[
\gamma_t : \sX_{[t] , [N]} \rightarrow \sU_{t, [N]}
\]
 with countable ranges.
 \end{definition}
  
\noindent Given a state process law $\mu \in \cP(\sX_{[T]})$ and an $N$-code $\gamma$, we let $\mathbf{P}_\mu^\gamma \in  \cP (\sX_{[T],[N]} \times \sU_{[T],[N]}) $ denote the induced strategic measure, i.e., joint probability law of the states observed by the IS and the actions generated by the DM:
\begin{align}
\begin{split}
& \mathbf{P}_\mu^\gamma (\d x_{[T] , [N]} , \d u_{[T] , [N]} )  \\
& = \displaystyle\bigotimes_{n \in [N]} \mu (\d x_{[T] , n}) \otimes \displaystyle\bigotimes_{t \in [T]} \delta_{\gamma_t(x_{[t] , [N]} ) }(\d u_{t , [N]} )    \label{eq:ULLN_5} .
\end{split}
\end{align}
We are interested in the \textit{minimum} information transmission rate needed by a sequential $N$-code in order to ensure that the realized sequence of states observed by the sender and actions taken by the decision-maker is $\Delta$-consistent (in expectation) with the target measure ${\Pr}_{\mu}^{\pi}$ on a fixed but arbitrary uGC class $\cF \subset M^1_b(\sX \times \sU)$. That is, we wish to design $\gamma$, so that
\begin{align}\label{eq:Delta_criterion}
\frac{1}{T}\sum^T_{t=1}\E^\gamma_\mu \left\| \sP_{X_{t,[N]},U_{t,[N]}}-\mu_t \otimes \pi_t \right\|_\cF \le \Delta,
\end{align}
while minimizing the total Shannon entropy of the messages $U_{1,[N]} = \gamma_1(X_{1,[N]}),\ldots,U_{T,[N]} = \gamma_T(X_{[T],[N]})$. Let
\begin{align}
\label{eq:rdrdrd}
\Gamma^N_{\mu , \pi} (\Delta) & \deq  \Big\{ \gamma = (\gamma_t)_{t \in [T]} : \nonumber\\
& \qquad  \frac{1}{T} \sum_{t \in [T]} \E^\gamma_\mu \left\Vert \mathrm{P}_{X_{t, [N]} , U_{t, [N]}} - \mu_t \otimes \pi_t \right\Vert_{\mathcal{F}} \le \Delta  \Big\}
\end{align}
be the set of all sequential $N$-codes that meet the criterion in \eqref{eq:Delta_criterion}. With this, we define the {\itshape operational sequential rate-distortion function for empirical coordination}:
\begin{equation}
\widehat{R}_{T,N}(\Delta)\triangleq\inf_{\gamma\in\Gamma^N_{\mu, \pi}(\Delta)}\frac{H\big(U_{[T],[N]}\big)}{NT}, \label{eq:opSRDF}
\end{equation}
where $H(U_{[T],[N]})$ is the joint Shannon entropy of the actions generated by the IR using $\gamma$.

{Our use of uGC classes in an operational criterion for coordination is inspired by the work of Al-Najjar \cite{alnajjar2009}, who analyzes the quality of forecasts or policy decisions made on the basis of estimating the probabilities of a whole class of events simultaneously from observed empirical frequencies. This amounts to evaluating the uniform deviation between the empirical probabilities and the `true' probabilities over a class $\cA$ of measurable sets. In order for the estimate to be consistent, the class of all indicator functions of the sets in $\cA$ must be a uGC class (which is equivalent to $\cA$ being a so-called Vapnik--Chervonenkis class of sets). Al-Najjar considers the case when the decision-makers have direct observation of all the relevant data. We are extending Al-Najjar's framework in three key ways:
\begin{itemize}
	\item We are considering arbitrary uGC classes, not just classes of indicator functions.
	\item We are imposing an information constraint (i.e., the state processes must be communicated to the DM over a finite-capacity channel).
	\item We are considering the sequential set-up, where, for each $n$, one must make $T > 1$ decisions, contingent on previously made decisions and the history of states. 
\end{itemize}}


\section{Main results}
\label{ssec:mmain}

Our main result addresses two questions pertaining to the operational rate-distortion function defined in \eqref{eq:opSRDF}:
\begin{enumerate}
	\item What is the minimum information transmission rate needed for IS to induce an empirical state-action distribution that is $\Delta$-consistent (in expectation) with the target measure ${\mathbf P}^\pi_\mu$?
	\item Can this minimum rate can be achieved by sequential $N$-codes?
\end{enumerate}
In order to address these questions, we first introduce an information-theoretic counterpart of \eqref{eq:opSRDF}, which we refer to as the {\itshape sequential rate-distortion function for empirical coordination}. 

Consider the subset of $\overrightarrow{\cM}(\sU_{[T]} | \sX_{[T]})$ consisting of those directed stochastic kernels whose induced marginal distributions of $(X_t , U_t)$ at each $t \in [T]$ are $\Delta$-consistent with $\mathbf{P}_\mu^\pi$, on average:
\begin{align*}
\label{eq:rdrd}
& \Pi_{\mu , \pi} (\Delta) \deq \Big\{ \tpi \in \overrightarrow{\cM}(\sU_{[T]} | \sX_{[T]}) :  \nonumber\\
& \quad \frac{1}{T} \sum_{t=1}^T \left\Vert \mu_t \otimes \tpi_t - \mu_t \otimes \pi_t \right\Vert_{\mathcal{F}} \le \Delta \Big\};
\end{align*}
here, given $\tpi \in \overrightarrow{\cM}(\sU_{[T]}|\sX_{[T]})$, $\tpi_t \in \cM(\sU|\sX)$ denotes the induced conditional distribution of $U_t$ given $X_t$. Then, the \textit{sequential rate-distortion function for empirical coordination} is defined as
\begin{equation}
\label{eq:rrrd}
R_{T}(\Delta) \deq \inf_{\tpi \in \Pi_{\mu , \pi}(\Delta)}\frac{I_{\mu,\tpi}(X_{[T]} ; U_{[T]})}{T}.
\end{equation}
\begin{remark} {\em For any $\tpi \in \overrightarrow{\cM}(\sU_{[T]} | \sX_{[T]})$, $U_t$ and $(X_{t+1} , \dots, X_T)$ are conditionally independent given $(X_{[t]} , U_{[t-1]})$ for each $t \in [T]$. Using this fact and the chain rule for mutual information, we can write
	\begin{align}
		I_{\mu,\tpi}(X_{[T]} ; U_{[T]}) &=\sum_{t \in [T]}  I_{\mu , \tpi} (X_{[T]} ; U_t | U_{[t-1]} ) \nonumber\\
		&=\sum_{t \in [T]}  I_{\mu ; \tpi} (X_{[t]}; U_t | U_{[t-1]}), \label{eq:dirinfo}
	\end{align}
where the quantity in \eqref{eq:dirinfo} is the {\em directed information} $I_{\mu,\tpi}(X_{[T]}\to U_{[T]})$ \cite{TatikondaThesis}. Thus, we can express the rate-distortion function in \eqref{eq:rrrd} as
\begin{align}\label{eq:rrrd_di}
	R_T(\Delta) = \inf_{\tpi \in \Pi_{\mu,\pi}(\Delta)} \frac{I_{\mu,\tpi}(X_{[T]}\to U_{[T]})}{T}.
\end{align}
Thus, $R_T(\Delta)$ is the empirical coordination counterpart of the sequential rate-distortion function \cite{TatikondaThesis,tatikondamitter2009feedback}. Equation \eqref{eq:rrrd_di} conveys an important intuitive concept, beyond the {\itshape common information content} embodied in the concept of mutual information. In a stochastic dynamical system, past states and actions convey information about the  current state. Each of the terms $I_{\mu , \tpi} (X_{[t]} ; U_t | U_{[t-1]})$ denotes the amount of information that needs to be conveyed about the state history $X_{[t]}$ beyond what is contained in $U_{[t-1]}$ in order to pin down the current action $U_t$.
\hfill$\square$}
\end{remark}

The rate-distortion function $R_{T}(\Delta)$ gives the smallest amount of information that any causal policy must convey about the sequence of states, on average per unit time, in order for the resulting joint measure to be $\Delta$-consistent (in expectation) with the postulated target measure ${\Pr}_{\mu}^{\pi}$ on the class $\cF$. Theorems~\ref{thm:mainn} and \ref{thm:converseee} below state that the sequential rate-distortion function for empirical coordination defined in \eqref{eq:rrrd} is the {\itshape asymptotic fundamental limit} of the empirical coordination problem formulated in Section~\ref{ssec:formulate}. {Note that the operational performance criterion in \eqref{eq:Delta_criterion} is \textit{non-additive} in $n$. Nevertheless, as evident from the two theorems below, the information-theoretic expression for the fundamental limit of sequential empirical coordination does not involve any limit as $N \to \infty$.}

\begin{theorem}[Achievability]
\label{thm:mainn}
Suppose $R_T(\Delta ) < \infty$. Then, for each $\varepsilon > 0$, there exists $ N(\varepsilon) \in \mathbb{N}$, such that 
\begin{equation*}
\widehat{R}_{T , N}(\Delta + \varepsilon ) \le R_T(\Delta) + \varepsilon.
\end{equation*}
In other words, under the conditions of the theorem, for each sufficiently large $N$, we can find a sequential $N$-code in $\Gamma_{\mu , \pi}^N(\Delta + \varepsilon)$, whose output entropy (normalized by $NT$) is approximately bounded by $ R_T(\Delta)$.
\end{theorem}

\begin{proof} {All of the heavy lifting needed in the proof is contained in two technical lemmas presented in Appendix~\ref{app:lemmas}. The key step is taken care of by Lemma~\ref{thm:quantization}, which extends the so-called Piggyback Coding Lemma \cite[Lemma~A.1]{raginsky2013empirical} to the sequential case. This lemma, in turn, relies on Lemma~\ref{lem:quant}, which provides a random coding argument along the lines of \cite[Lemma~9.3.1]{gallager1968information} for tree codes (a natural choice in the presence of causality constraints).  With these two lemmas at hand, the achievability proof is conceptually transparent.}
	
Since $R_T(\Delta) < \infty$, there exists some $\tpi \in \Pi_{\mu , \pi} (\Delta + \frac{\varepsilon}{2})$ such that
\[
\frac{1}{T}  I_{\mu ; \tpi}(X_{[T]} ; U_{[T]}) < R_T(\Delta) + \frac{\varepsilon}{2}.
\]
For each $t$, define the function
 \[
 \psi_{t , N}(x_{t,[N]};u_{t,[N]}) \deq  \left\Vert \sP_{x_{t,[N]},u_{t ,  [N]}} - \mu_t \otimes \tpi_t \right\Vert_{\mathcal{F}}.
\] 
Since $\cF$ is a uGC class and since $(X_{t,1},U_{t,1}),\ldots,(X_{t,N},U_{t,N})$ are i.i.d.\ under ${\mathbf P}^{\tpi}_\mu$, we have
\begin{equation*}
\displaystyle\max_{t \in [T]} \lim_{N \rightarrow \infty} \E_{\mu}^{\tpi} \left[ \psi_{t,N} (X_{t,[N]},U_{t ,  [N]}) \right] = 0.
\end{equation*}
Then by Lemma \ref{thm:quantization} in the appendix, there exists a sequential $N$-code $\gamma \in \Gamma^N_{\mu , \tpi} ( \frac{\varepsilon}{2})$, such that, for $U_{[T],[N]} = \big(\gamma_t (X_{[t] , [N]})\big)_{t \in T}$, we have
\begin{align*}
 \frac{H(U_{[T],[N]})}{NT} & \le \frac{1}{T} I_{\mu , \tpi}(X_{[T]} ; U_{[T]})  + \frac{\varepsilon}{2}  < R_T(\Delta) + \varepsilon .
\end{align*}
Moreover, using the triangle inequality, we can estimate
\begin{align*}
& \frac{1}{T} \sum_{t \in [T]} \E^{\gamma}_{\mu} \left\Vert \sP_{X_{t,[N]} , U_{t, [N]} } - \mu_t \otimes \pi_t \right\Vert_{\mathcal{F}} \\
 &\le   \frac{1}{T} \sum_{t \in [T]} \E^{\gamma}_\mu\left\Vert \sP_{X_{t,[N]} , U_{t, [N]}}  - \mu_t \otimes \tpi_t \right\Vert_{\mathcal{F}} + \frac{1}{T} \sum_{t \in [T]} \left\Vert \mu_t \otimes \tpi_t - \mu_t \otimes \pi_t \right\Vert_{\mathcal{F}} \\
   &\le \Delta + \varepsilon .
\end{align*}
Thus, $\gamma \in \Gamma^N_{\mu,\pi}(\Delta+\eps)$, and therefore, from the definition of $\widehat{R}_{T,N}(\cdot)$, it follows that
\[
\widehat{R}_{T,N} (\Delta + \varepsilon) \le \frac{H(U_{[T],[N]})}{NT}  \le  R_T(\Delta) + \varepsilon .
\]

\end{proof}

\begin{theorem}[Converse]
\label{thm:converseee}
For any $N$, $T$ and $\Delta$,  
\begin{equation*}
\widehat{R}_{T , N}(\Delta ) \ge R_{T }(\Delta).
\end{equation*}
In other words, the average output entropy of any $N$-code $\gamma \in \Gamma^N_{\mu, \pi} (\Delta)$ must be at least as large as $R_T(\Delta)$.
\end{theorem}

\begin{proof}The proof uses the techniques from \cite{raginsky2013empirical}. Fix an arbitrary sequential $N$-code $\gamma \in \Gamma^{N}_{\mu, \pi} (\Delta)$, and let $(X_{[T] , [N]} , U_{[T] , [N]} )$ the state-action process with process law $\Pr_\mu^{\gamma}$. Let $J$ be a random variable uniformly distributed on $[N]$, independently of $(X_{[T] , [N]} ,U_{[T] , [N]})$. Consider the random couple $( X_{[T] , J} , U_{[T] , J} )$. From symmetry and independence, it follows that the marginal distribution of $X_{[T] , J}$ is equal to $\mu$. For each $t \in [T]$, let $\tpi^{(t)} \in \cM(U| X_{[t]} \times U_{[t-1]})$ be the induced conditional law of $U_{t,J}$ given $( X_{[t],J} , U_{[t-1],J} )$, and let $\tpi_t \in \cM(\sU | \sX)$ denote the induced conditional law of $U_{t,J}$ given $X_{t,J}$. Then we have the following chain of equalities and inequalities: 
\begin{equation*}
\begin{split}
 H \left( (\gamma_t(X_{[t],[N]}) )_{t \in [T]} \right) = & H(U_{[T], [N]}) \\
  \overset{\mathrm{(a)}}{=}   & I(X_{[T], [N]} ; U_{[T], [N]} )\\
 \overset{\mathrm{(b)}}{\ge} & \sum_{n \in [N]} I(X_{[T], n} ; U_{[T] , n}) \\
\overset{\mathrm{(c)}}{=} &  NI(X_{[T] , J} ; U_{[T] , J} |J) \\
\overset{\mathrm{(d)}}{=} & NI(X_{[T] , J} ; U_{[T] , J} , J) \\
 \ge & N I_{\mu , \tpi} (X_{[T]} ; U_{[T]}) ,
\end{split}
\end{equation*}
where:
\begin{itemize}
\item (a) follows from the fact that the map $X_{[T],[N]} \rightarrow U_{[T],[N]}$
is deterministic;
\item (b) is a standard information-theoretic fact: if $X_{[N]}$ is a sequence of independent random variables, then for any sequence $Y_{[N]}$ of random variables jointly distributed with the $X_{n}$'s, 
\[
\sum_{n \in [N]}I(X_{n}; Y_{n}) \le I(X_{[N]};Y_{[N]});
\]

\item (c) follows from the construction of $J$;
\item (d) follows from the fact that, since $\{X_{[T],1}, \ldots , X_{[T],N} \}$ are i.i.d., $J$ and $X_{[T],J}$ are independent (see Appendix~B
in \cite{raginsky2013empirical}), and from the chain rule for the mutual infornation.
\end{itemize}
The remaining steps are consequences of definitions and of standard information-theoretic identities. Dividing both sides by $NT$, we obtain the bound 
$$
 \frac{I_{\mu , \tpi} (X_{[T]} ; U_{[T]})}{T} \le \frac{H \left( U_{[T],[N]} \right)}{NT}. 
$$
Now, for each $t \in [T]$, $X_{t,J}$ is independent of $J$, and has the same law as $X_{t,1}$, namely $\mu_t$. Moreover, (cf.~Appendix~B in \cite{raginsky2013empirical}),
the expected empirical distribution ${\mathbb{E}}_{\mu}^{\gamma} \sP_{X_{t,[N]}, U_{t,[N]}}$
is equal to $\mu_{t} \otimes \tpi_{t}$. Then we have 
\begin{equation*}
\begin{split}
&\sum_{t \in [T]}\left\Vert \mu_{t}\otimes\pi_{t}-\mu_{t}\otimes\tpi_{t}\right\Vert _{\mathcal{F}} \\
& = \sum_{t \in [T]}\left\Vert {\mathbb{E}}_{\mu}^{\gamma} \sP_{X_{t,[N]}, U_{t,[N]}}-\mu_{t}\otimes\tpi_{t}\right\Vert _{\mathcal{F}}\\
 & \leq \sum_{t \in [T]}{\mathbb{E}}_{\mu}^{\gamma} \left\Vert \sP_{X_{t,[N]}, U_{t,[N]}}-\mu_{t}\otimes\tpi_{t}\right\Vert _{\mathcal{F}}\\
 & \leq\Delta,
\end{split}
\end{equation*}
where the first inequality is by convexity, while the second inequality is by assumption on $\gamma$. Therefore, $\tpi = \left(\tpi^{(t)} \right)_{t \in [T]} \in \Pi_{\mu, \pi} (\Delta)$, and consequently
\[
R_T(\Delta) \leq \frac{I_{\mu , \tpi} (X_{[T]} ; U_{[T]})}{T} \leq \frac{H \left( U_{[T],[N]}\right)}{NT} ,
\]
by definition. Since this holds for every $\gamma \in \Pi^N_{\mu , \pi} (\Delta)$, it follows that $ R_T(\Delta) \le \widehat{R}_{T,N} (\Delta) $. \end{proof}

 \section{Examples and bounds}
 \label{sec:implicc}

Although Theorems \ref{thm:mainn} and \ref{thm:converseee} provide a full characterization of the fundamental limits on the minimal rate of communication for sequential empirical coordination, the computation of the sequential rate distortion function $R_T(\Delta)$ is a complicated  optimization problem already in the static $(T = 1)$ case, which was addressed in \cite{raginsky2013empirical}. Below, we provide two examples that illustrate the difficulty of explicitly computing $R_T(\Delta)$ even for $T = 1$. {We also show that, in some cases, one can upper-bound $R_T(\Delta)$ by a simpler information-theoretic quantity related to remote lossy source coding.}

\subsection{ Kolmogorov-Smirnov criterion for one-step costs}
\label{sec:kolmog}

While we have remained silent on the nature of the target strategic measure $\Pr_\mu^\pi$, it may have been selected based on considerations of expected cost. Thus, suppose that we have a function $c \in M(\sX \times \sU)$, such that $c(x, u)$ gives the cost of taking action $u$ in response to state $x$. Let $\mathcal{F}$ be the class of indicator functions of the level sets of $c$:
\begin{equation}
\label{eq:indic}
f_a(x,u) \deq \1 \{c(x, u) \le a \}, \quad a \in \mathbb{R} .
\end{equation}
Then we have the following:

\begin{proposition}
Let $\cF$ denote the class of all $f$ of the form \eqref{eq:indic}. Then for any two $P , Q \in \cP(\sX \times \sU)$,
\begin{equation}
\left\Vert P - Q \right\Vert_{\cF} = d_{\rm KS} (F_{P \circ c^{-1}} , F_{Q \circ c^{-1}}),
\end{equation}
where $F_\mu$ denotes the cumulative distribution function (cdf) of a Borel probability measure $\mu$ on the reals, and
\begin{equation}
 d_{\rm KS} (F , F^{'}) \deq \sup_{a \in \mathbb{R}} | F(a) - F^{'}(a)|
\end{equation}
is the Kolomogorov-Smirnov distance between cdf's $F$ and $F^{'}$. The class $\cF$ is a universal Glivenko-Cantelli class.
\end{proposition}

\begin{proof}
Fix any pair $P, Q \in \cP(\sX \times \sU)$. Then the chain of equalities
\begin{align*}
\left\Vert P - Q \right\Vert_{\cF} =& \sup_{a \in \mathbb{R}} |P[c(\sX , \sU) \le a] - Q[c(X , U) \le a] | \\
=& \sup_{a \in \mathbb{R}} | F_{P \circ c^{-1}} (a) -  F_{Q \circ c^{-1}} (a)| \\
=& d_{\rm KS} (F_{P \circ c^{-1}},  F_{Q \circ c^{-1}})
\end{align*}
follows from definitions. By the classical Glivenko-Cantelli theorem [14, Prop. 4.24], the class of all indicator functions $r \mapsto \1\{r \le a\}, a \in \mathbb{R}$, is a uGC class on $(\mathbb{R},\mathcal{B}(\mathbb{R}))$. Therefore, since $\{P \circ c^{-1} : P \in \cP(\sX \times \sU)\} \subset \cP(\mathbb{R})$, $\cF$ is a uGC class of functions on $\sX \times \sU$.
\end{proof}

The following is immediate from the above proposition:

\begin{theorem}
\begin{equation}
\begin{split}
R_T(\varepsilon) = & \inf_{\tpi \in \overrightarrow{\cM}(\sU^T | \sX^T) } \left\{ I_{\mu , \tpi} (X_{[T]} ; U_{[T]}) :  \right. \\ 
& \left. \frac{1}{T} \sum_{t=1}^T d_{\rm KS} (F_{(\mu_t \otimes \tpi_t) \circ c^{-1}}  , F_{(\mu_t \otimes \pi_t) \circ c^{-1}}  ) \le \Delta \right\}.
\end{split}
\end{equation}
\end{theorem}
\noindent In other words, $R_T(\Delta)$ is the smallest mutual information between the state process $X_{[T]}$ with law $ \mu \in \cP(\sX_{[T]})$ and any action process $U_{[T]}$ generated from  $X_{[T]}$ by a causal policy  $\tpi$, such that the time average of the Kolmogorov-Smirnov distances between the state-action costs under $\tpi$ and the target policy $\pi$ is bounded from above by $\Delta$. Evaluating this quantity exactly is difficult even for $T = 1$.

\subsection{Weak convergence and Wasserstein distances}

Another example concerns approximation of the target strategic measure $\Pr_{\mu}^\pi$ in a certain metric that metrizes the topology of weak convergence of probability measures. Suppose that $\sX \times \sU$ is a Polish space with a given metric $d$. For any $f \in M(\sX \times \sU)$, define the {\itshape Lipschitz norm}
\begin{equation}
\left\Vert f \right\Vert_{\rm Lip} \deq \sup_{(x , u) , (y , v) \in \sX \times \sU} \frac{|f(x,u) - f(y, v)|}{d((x,u) , (y,v))}
\end{equation}
and the bounded Lipschitz norm
\begin{equation}
\left\Vert f \right\Vert_{\rm BL} \deq \left\Vert f \right\Vert_{\infty} + \left\Vert f \right\Vert_{\rm Lip}
\end{equation}

\begin{proposition}
Consider the function class $\cF = \{f \in M(\sX \times \sU) :  \left\Vert f \right\Vert_{\rm BL} \le 1 \}$. Then, for any two $P , Q \in \cP(\sX \times \sU)$,
\begin{equation}
\label{eq:ddist}
||P - Q||_{\cF} = d_{\rm BL} (P , Q),
\end{equation}
the bounded Lipschitz metric on $\cP(\sX \times \sU)$ that metrizes the topology of weak convergence of probability measures. The class $\cF$ is a universal Glivenko-Cantelli class.
\end{proposition}

\begin{proof}
\sloppypar Eq.~\eqref{eq:ddist} is the definition of the bounded Lipschitz metric \cite[Sec.~11.3]{dudley2002real}, which metrizes the topology of weak convergence of probability measures \cite[Thm.~11.3.3]{dudley2002real}. Now, let $(X_1, U_1), (X_2, U_2), . . .$ be a sequence of i.i.d. random elements of $\sX \times \sU$ with common marginal law $\Pr$ . Then 
\begin{equation}
d_{BL}(\sP_{(X_{[n]},U_{[n]})},\Pr) \xrightarrow{n \rightarrow \infty} 0, ~~ \Pr-{\rm a.s.}
\end{equation}
by Varadarajan's theorem \cite[Thm.~11.4.1]{dudley2002real}. Since this holds for any $\Pr \in \cP(\sX)$, and in light of \eqref{eq:ddist}, we conclude that $\cF$ is a uGC class.
\end{proof}

Under an additional moment condition, the bounded Lipschitz metric can be upper-bounded by the so-called Wasserstein metric. Let $\cP_0(\sX \times \sU)\subset \cP (\sX \times \sU)$ be the set of all probability measures $\cP$ for which there exists some $(x_0,u_0) \in  \sX \times \sU$, such that $\left\langle \Pr , d(\cdot,(x_0,u_0)) \right\rangle < \infty$. The Wasserstein metric between any two $P, Q \in \cP_0(\sX \times \sU)$ is
\begin{equation}
W_d(P,Q) \deq \sup_{ \left\Vert f \right\Vert_{\rm Lip} \le 1} | \langle P,f \rangle - \langle Q,f \rangle | .
\end{equation}
We can now give the following upper bound on the sequential rate-distortion function $R_T (\Delta)$ w.r.t. $\cF$:

\begin{theorem}
Suppose that $\Pr_\mu^\pi \in \cP_0 (\sX \times \sU)$. Then
\begin{equation}
\begin{split}
R_T(\Delta) \le & \inf_{\tpi \in \overrightarrow{\cM}(\sU_{[T]}| \sX_{[T]})} \left\{ I_{\mu , \tpi} (X_{[T]} ; U_{[T]}) : \right. \\
& ~~~ \left.  \frac{1}{T} \sum_{t \in [T]} W_d(\mu_t \otimes \tpi  ,\mu_t \otimes \pi_t) \le \Delta  \right\}.
\end{split}
\end{equation}
\end{theorem}
\noindent Again, despite the clean conceptual interpretation of $R_T(\Delta)$ in terms of approximating strategic measures by empirical distributions of state-action pairs under the bounded Lipschitz metric, it does not admit closed-form expressions even for $T=1$.

\subsection{Upper bounds on the sequential rate-distortion function}

{While the exact computation of $R_T(\Delta)$ is a difficult task, it is possible to obtain computable upper bounds under some additional regularity assumptions. One example is given in the theorem below. To keep things simple, we consider the case of $T=1$.}
{\begin{theorem} Suppose that there exists a metric $d$ on the action space $\sU$, such that the elements of $\cF$ satisfy the following uniform Lipschitz condition: for all $u,u' \in \sU$ and all $f \in \cF$,
\begin{align}\label{eq:uniform_Lip}
	\sup_{x \in \sX} |f(x,u)-f(x,u')| \le d(u,u').
\end{align}
Then $R_1(\Delta) \le \overline{R}(\Delta)$, where
\begin{align}\label{eq:KOP}
\overline{R}(\Delta) := \inf_{P_{\widehat{U}|U} \in \cM(\sU|\sU): \atop \E[d(U,\widehat{U})]\le\Delta}  I(X; \widehat{U}).
\end{align}
\end{theorem}}
{
\begin{remark}{\em The function defined in \eqref{eq:KOP} has been introduced in a recent paper of Kochman, Ordentlich, and Polyanskiy \cite{kop2018ozarow} in the context of converse bounds for multiple-description source coding and joint source-channel broadcasting of a common source.}\end{remark}}
{\begin{proof} Disintegrate the joint probability law of $(X,U)$ as $\mu \otimes \pi$, where $\mu \in \cP(\sX)$ and $\pi \in \cM(\sU|\sX)$. Fix a Markov kernel $P_{\widehat{U}|U} \in \cM(\sU|\sU)$ satisfying $\E[d(U,\widehat{U})] \le \Delta$ and define $\tilde{\pi} \in \cM(\sU|\sX)$ by
\begin{align}\label{eq:tilde_pi_testch}
	\tilde{\pi}(\cdot|x) := \int_{\sU} \pi(\d u |x) P_{\widehat{U}|U}(\cdot|u).
\end{align}
Then $X$ and $\widehat{U}$ are conditionally independent given $U$. Using this fact and the uniform Lipschitz property \eqref{eq:uniform_Lip}, we have for any $f \in \cF$
\begin{align*}
&\langle \mu \otimes \pi, f\rangle - \langle \mu \otimes \tilde{\pi}, f\rangle \\
&\qquad=  \E\big[\E[f(X,U)-f(X,\widehat{U})|X,U]\big] \\
	&\qquad\le \E\big[\E[d(U,\widehat{U})|X,U]\big]  \\
	&\qquad= \E [d(U,\widehat{U})] \\
	&\qquad\le \Delta.
\end{align*}
Interchanging the roles of $\pi$ and $\tilde{\pi}$, we obtain
$$
|\langle \mu \otimes \pi, f \rangle - \langle \mu \otimes \tilde{\pi},f\rangle| \le \Delta.
$$
Taking the supremum over all $f \in \cF$, we see that $\|\mu \otimes \pi - \mu \otimes \tilde{\pi}\|_\cF \le \Delta$. Optimizing over all such $P_{\widehat{U}|U}$, we get the bound $R_1(\Delta) \le \overline{R}(\Delta)$.
\end{proof}}
{As an illustration, consider the case when $\sX$ and $\mu$ are arbitrary, $\sU = \Reals$, and the policy $\pi$ is deterministic: $\pi(\d u |x) = \delta_{g(x)}(\d u)$ for some Borel function $g : \sX \to \Reals$. Suppose, furthermore, that the uniform Lipschitz condition \eqref{eq:uniform_Lip} is satisfied with $d(u,u') = |u-u'|$. Then we have the following:
\begin{itemize}
	\item If $\sqrt{\E [g^2(X)]} = m <\infty$, then
	$$
	R_1(\Delta) \le C_{\rm av}\left(\frac{m^2}{\Delta^2}\right),
	$$
	where, for $s \ge 0$,
	\begin{align*}
	C_{\rm av}(s) &\deq \sup_{Y:\, {\rm var}[Y]\le 1} I(Y; \sqrt{s} Y + Z) \\
	&= \frac{1}{2}\log(1+s)
	\end{align*}
	is the Shannon capacity of the additive white Gaussian noise (AWGN) channel under the average power constraint (the additive noise $Z$ is a standard normal random variable independent of $Y$).
	\item If $\|g\|_\infty = m < \infty$, then
	$$
	R_1(\Delta) \le C_{\rm pk}\left(\frac{m^2}{\Delta^2}\right),
	$$
where
\begin{align*}
	C_{\rm pk}(s) \deq \sup_{Y:\, |Y| \le 1 \text{ a.s.}} I(Y; \sqrt{s} Y+Z)
\end{align*}
is the Shannon capacity of the AWGN channel under the peak power constraint.
\end{itemize}
To derive both of these bounds, the natural choice of $P_{\widehat{U}|U}$ is given by the additive Gaussian noise channel $\widehat{U} = U+\Delta Z = g(X) + \Delta Z$. Then the Markov kernel $\tilde{\pi}$ defined in \eqref{eq:tilde_pi_testch} evidently satisfies
$$
\| \mu \otimes \pi - \mu \otimes \tilde{\pi}\|_\cF \le \Delta,
$$
and
$$
\overline{R}(\Delta) \le I(X; U + \Delta Z) = I(X; g(X)+\Delta Z).
$$
Since $X$ and $\widehat{U}= g(X) + \Delta Z$ are conditionally independent given $U = g(X)$, we have
$$
I(X; g(X) + \Delta Z) = I(g(X) ; g(X) + \Delta Z) = I(U; U + \Delta Z).
$$
Thus, in the case $m^2 = \E U^2 < \infty$,
\begin{align*}
	I(X; \widehat{U}) &= I(U ; U + \Delta Z) \\
	&\le \sup_{U:\, {\rm var}[U] \le m^2} I(U ; U + \Delta Z) \\
	&= C_{\rm av}\left(m^2/\Delta^2\right).
\end{align*}
Similarly, if $|U| \le m$ a.s., then
\begin{align*}
	I(X; \widehat{U}) &= I(U; U + \Delta Z) \\
	&\le \sup_{U \in [-m,m] \text{ a.s.}} I(U; U + \Delta Z) \\
	&= C_{\rm pk}\left(m^2/\Delta^2\right).
\end{align*}
}

\begin{appendices}
\renewcommand{\theequation}{\Alph{section}.\arabic{equation}}
\renewcommand{\thelemma}{\Alph{section}.\arabic{lemma}}
\renewcommand{\theproposition}{\Alph{section}.\arabic{proposition}}
\renewcommand{\thedefinition}{\Alph{section}.\arabic{definition}}

\section{Universal Glivenko--Cantelli classes and typical sequences in standard Borel spaces}
\label{app:uGCappendix}

	\setcounter{lemma}{0}
	\setcounter{definition}{0}
	\setcounter{proposition}{0}
	\setcounter{equation}{0}
	
 If $X_1,\ldots,X_N$ are i.i.d.\ random elements of $\sX$ with common marginal law $\mu$, then for any $f \in M(\sX)$ the empirical means
 \[
 \langle \sP_{X_{[N]}} , f \rangle = \frac{1}{N}\sum_{n \in [N]} f(X_n), \quad N \in \mathbb{N} 
 \]
converge to the mean $\langle \mu , f \rangle$ almost surely, by the Strong Law of Large Numbers (SLLN). By the union bound, this statement carries over to any finite family of functions. Thus, for any $\cF \subset M(\sX)$  with $|\cF| < \infty$,
\begin{equation}
\label{eq:metricconvergence}
\left\Vert\sP_{X_{[N]}} - \mu\right\Vert_{\mathcal{F}} \xrightarrow{N\to\infty}0,\qquad\mu\text{-a.s.} 
\end{equation}
In general, \eqref{eq:metricconvergence} is referred to as the {\itshape Uniform Law of Large Numbers} (ULLN) over $\cF$ -- that is, the worst-case absolute deviation between empirical and true means converges to zero {\itshape uniformly} over the function class $\cF$.  However, ULLN may not hold for an arbitrary infinite class of functions $\cF$ on a general Borel space. Specifically, it fails to hold on $\mathcal{F} \equiv \{f \in M(\sX) : \|f\|_\infty \le 1\}$ if $\mu$ has a density.

This observation shows that properly defining the notion of a typical sequence over an abstract Borel alphabet requires some care. Let us recall the usual definition:
\begin{definition} \label{def:strongtypp} Given a finite set $\sX$ and a probability distribution $\mu \in \cP(\sX)$ on it, the {\itshape typical set} $\mathcal{T}_\Delta^{(N)}(\mu)$, for $\Delta > 0$, is the set of all $N$-tuples $x_{[N]} \in \sX_{[N]}$ whose empirical distributions $\sP_{X_{[N]}}$ are $\Delta$-close to $\mu$ in the total variation norm:
 \[
 \mathcal{T}_\Delta^{(N)}(\mu) \deq \left\{ x_{[N]} \in \sX_{[N]} : \left\Vert \sP_{X_{[N]}}-\mu\right\Vert_{\rm TV} < \Delta \right\}.
 \]
\end{definition}
\noindent 
Now, if $\mathcal{F} \equiv \{f : \|f\|_\infty \le 1\}$, then $\Vert \cdot \Vert_\cF$ coincides with the total variation norm
\begin{equation}
\label{eq:totalvar}
\Vert\nu\Vert_{\rm TV} \deq 2\sup_{A \in \mathcal{B}(\sX) }\left\vert \nu(A) \right\vert .
\end{equation}
Therefore, we have the following implication of \eqref{eq:metricconvergence} with $\cF = \{f : \|f\|_\infty \le 1\}$: If $X_1,X_2,\ldots$ are i.i.d.\ elements of a finite alphabet $\sX$ with common marginal $\mu$, then
 \begin{equation} 
 \label{eq:typp}
  \mathbf{P} \left(X_{[N]} \notin \mathcal{T}_\Delta^{(N)}(\mu) \right) \xrightarrow{N\to\infty}0.
 \end{equation}
In order to extend the intuitive notion of typicality to general Borel alphabets, we restrict the class $\cF$ to be a universal Glivenko--Cantelli class. Now, typical sequences on general Borel spaces can be defined in the spirit of Definition \ref{def:strongtypp}:
 
\begin{definition} \label{def:typuGC} Fix a uGC function class $\cF$ on $\sX$. Given a probability measure $\mu \in \cP(\sX)$, the {\itshape typical set} $\mathcal{T}_{\Delta , \cF}^{(N)}(\mu)$, for $\Delta > 0$, is the set of all $N$-tuples $x_{[N]} \in \sX_{[N]}$ whose empirical distributions $\sP_{X_{[N]}}$ are $\Delta$-close to $\mu$ in the $\left\Vert \cdot \right\Vert_{\cF} $ seminorm:
 \[
 \mathcal{T}_{\Delta , \cF}^{(N)}(\mu) \deq \left\{ x_{[N]} \in \sX_{[N]} : \left\Vert \sP_{X_{[N]}}-\mu\right\Vert_{\cF} < \Delta \right\}.
 \]
\end{definition}
\noindent In other words, the typical set $T_{\Delta,\cF}^{(N)}(\mu)$ consists of all $x_{[N]}$, whose empirical distributions are $\Delta$-consistent with $\mu$ on the class $\cF$. We then have the following counterpart of \eqref{eq:typp}:

\begin{proposition} Consider a Borel space $\sX$ and a uGC class $\cF \subset \{f \in M(\sX) : \|f\|_\infty \le 1\}$. If $X_1,X_2,\ldots$ is a sequence of i.i.d.\ random elements of $\sX$ with common law $\mu$, then for any $\Delta > 0 $
 \begin{equation}
 \mathbf{P} (X_{[N]} \notin  \mathcal{T}_{\Delta , \cF}^{(N)}(\mu) ) \xrightarrow{N\to\infty}0.
 \end{equation}
\end{proposition} 

\begin{proof}
Immediate from Definitions \ref{def:uGC} and \ref{def:typuGC}.
\end{proof}

\section{Technical lemmas for the proof of Theorem~\ref{thm:mainn}}
\label{app:lemmas}

	\setcounter{lemma}{0}
	\setcounter{equation}{0}

Lemma \ref{thm:quantization} below is at the heart of the proof of Theorem \ref{thm:mainn}. It states that, for any sequence of functions on $N$-blocks of state-action pairs whose expected values vanish asymptotically under a given strategic measure, one may construct a sequence of sequential $N$-codes, under which the expected value of the time-average of these functions can be made arbitrarily small, and whose output entropy is upper-bounded by the mutual information of the source and action under the given probability measure.

\begin{lemma} \label{thm:quantization} Consider a pair $(\mu, \tpi) \in \cP(\sX_{[T]})\times\overrightarrow{\cM}(\sU_{[T]} | \sX_{[T]})$, such that  $I_{\mu, \tpi}(X_{[t]} ; U_t | U_{[t-1]}) < \infty$ for each $t \in [T]$. Let $(X_{[T] , n} , U_{[T] , n})_{n \in [N]}$ be $N$ i.i.d.\ copies of the state-action processes with process law $\Pr_\mu^{\tpi} $. Let $\psi_{t,N}$ be a sequence of bounded measurable functions $ \psi_{t,N} : \sX_{t, [N]} \times \sU_{t, [N]} \rightarrow [0,1]$ obeying
\begin{equation*}
 \begin{split}
  &\lim_{N \rightarrow \infty} \E_{\mu}^{\tpi} \left[  \psi_{t, N}   (X_{t, [N]} , U_{t , [N]})\right] = 0 .
   \end{split}
\end{equation*}
Then, for any $\varepsilon > 0$, there exists $N_0 = N_0(\varepsilon)$, such that, for every $N >N_0 $, we can find $T$ mappings $ \gamma_{t,N}: \sX_{[t] , [N]} \rightarrow \sU_{[N]} $, $t \in [T]$, satisfying
\begin{equation}
  \label{eq:functionbound}
\frac{1}{T}\sum_{t \in [T]}  \E_{\mu} \left[ \psi_{t,N} (X_{t, [N]} , \gamma_{t,N}(X_{[t], [N]}) ) \right] \le \varepsilon ,
\end{equation}
while
\begin{equation}
 \label{eq:infobound}
\frac{1}{NT} H\left( \Big\{  \gamma_{t,N} (X_{[t] , [N]})  \Big\}_{t \in [T]} \right)  \le \frac{1}{T} I_{\mu , \tpi}(X_{[T]} ; U_{[T]}) + \varepsilon .
\end{equation}
\end{lemma}

\begin{proof}

Let $\delta_{t,N} := \E_{\mu}^{\tpi} \left[  \psi_{t, N}    (X_{t,[N]} , U_{t,[N]}) \right]$ and define the set
\begin{align}
\label{eq:atild}
A_{t,N} &\deq \Big\{(x_{t,[N]} , u_{t,[N]}) : \psi_{t,N} (x_{t,[N]} , u_{t,[N]})  \le \sqrt{ \delta_{t,N}}  \Big\}.
\end{align}
By Markov's inequality, 
\begin{equation*}
\Pr_{\mu}^{\tpi}[A^{c}_{t,N} ] \le \frac{\delta_{t,N}}{\sqrt{\delta_{t,N}}  } = \sqrt{\delta_{t,N}} \xrightarrow{N \to \infty} 0.
\end{equation*}
This implies that, for each $t \in [T]$, the state-action tuple $(X_{t,[N]},U_{t,[N]})$ generated according to $(\mu,\tpi)$ belongs to $A_{t,N}$ with high probability. By Lemma \ref{lem:quant} below, for all large enough $N$, there exist measurable mappings $\gamma_{t,[N]} : \sX_{[t],[N]} \to \sU_{t,[N]}$, $t \in [T]$, such that\footnote{By adjusting $\varepsilon$, we can ensure that $M_t$ is an integer.}
$$
M_t \deq \left|\gamma_{t,[N]}(\sX_{[t],[N]})\right| -1 = e^{N(R_t+\varepsilon)}, 
$$
where $R_t = I_{\mu,\tpi}(X_{[t]};U_t | U_{[t-1]})$, and
\begin{align*}
&\max_{t \in [T]}  \Pr_\mu \left[ (X_{t, [N]} , \gamma_t(X_{[t],[N]})) \not\in  A_{t,N} \right]\\
&\le T ~ \max_{t \in [T]} ~ \Big\{ \Pr_{\mu}^{\tpi} \left[   (X_{t, [N]} ,  U_{t,[N]} ) \not\in A_{t,N} \right]  \\
& \hspace{50pt}  +\Pr_{\mu}^{\tpi} \left[ i_t(X_{[t],[N]} ; U_{[t],[N]} ) > N (R_t + \varepsilon / 2 ) \right]  \\
& \hspace{50pt} + \exp{(-M_t e^{-N (R_t + \varepsilon / 2) } )} \Big\} \\
& \le T \exp(-e^{N \varepsilon /2}) + T \max_{t \in [T]} \Big\{  \sqrt{ \delta_{t,N}} \\
& \qquad \qquad + \Pr_{\mu}^{\tpi} \left[ i_t(X_{[t],[N]} ; U_{[t],[N]} ) > N (R_t + \varepsilon / 2 ) \right] \Big\},
\end{align*}
where
\begin{align*}
& i_s(x_{[s],[N]};u_{[s],[N]}) \\
&\deq \log \frac{\mathrm{d} \Pr^{\mu, \tpi}_{U_{s , [N]} | (X_{[s],[N]}, U_{[s-1],[N]}) = ( x_{[s],[N]} , u_{[s-1],[N]}) }  }{\mathrm{d} \Pr^{\mu, \tpi}_{U_{s , [N]} | U_{[s-1],[N]}  = u_{[s-1],[N]} } }(u_{s,[N]}).
\end{align*}
are the conditional information densities.
Now, since $\psi_{t,N}$ takes values in $[0,1]$, we can write
\begin{align}
&\frac{1}{T} \sum_{t \in [T]} \E_{\mu} \left[  \psi_{t,N} (X_{t,[N]} , \gamma_{t,N}(X_{t,[N]}) ) \right]\nonumber \\
&  \le ~  \max_{t \in [T]} \Pr_{\mu}[  (X_{t, [N]} , \gamma_t(X_{[t],[N]})) \not\in A_{t,N}   ]  + \max_{t \in [T]} \sqrt{ \delta_{t,N}} ,\label{eq:summation}
\end{align}
For all sufficiently large $N$, the right hand side of \eqref{eq:summation} can be made smaller than $\varepsilon$. To see this, notice that $\displaystyle\max_{t \in [T]} \delta_{t,N} \rightarrow 0 $ as $N \rightarrow \infty$ by assumption. Moreover, since
\begin{align*}
 & i_t(X_{[t] , [N]} ; U_{[t] , [N]}) \\
 & \deq \log \frac{\mathrm{d} \Pr^{\mu, \tpi}_{U_{t , [N]} | X_{[t],[N]} , U_{[t-1],[N]} }  }{\mathrm{d} \Pr^{\mu, \tpi}_{U_{t , [N]} | U_{[t-1],[N]} }  } (X_{[t] , [N]} ,  U_{[t], [N]})\\
 &= \sum_{n \in [N]} \log \frac{\mathrm{d} \Pr^{\mu, \tpi}_{U_{t , n} | X_{[t],n} , U_{[t-1],n} }  }{\mathrm{d} \Pr^{\mu, \tpi}_{U_{t , n} | U_{[t-1],n} }  } (X_{[t] , n} ,  U_{[t], n}) ,
\end{align*}
is a sum of i.i.d.\ random variables, and 
\begin{align*}
 R_t = &I_{\mu , \tpi}(X_{[t]} ; U_t  | U_{[t-1]})   =  \frac{1}{N} \E_{\mu}^{\tpi} \left[  i_t (X_{[t] , [N]} ; U_{[t] , [N]}) \right],
\end{align*}
the quantity 
\[
\max_{t \in [T]} \Pr_\mu^{\tpi} \left[  i_t(X_{[t] , [N]} ; U_{[t] , [N]}) > N (R_t + \varepsilon/2) \right]
\]
can be made as small as desired for all large $N$, according to the law of large numbers. Therefore, we can find a sufficiently large $N_0 = N(\varepsilon)$ and a sequential code $\gamma = (\gamma_{t,N})_{t \in [T]}$, such that Eq.~\eqref{eq:functionbound} holds.

 Towards verifying Eq.~\eqref{eq:infobound}, let $\Pr_\mu^\gamma$ be the joint probability law of $X \sim \mu$ and the output of $\gamma$. For $(X, U) \sim \Pr_\mu^\gamma$, we have
\begin{equation*}
\begin{split}
& \frac{1}{NT} H\left( \Big\{  \gamma_{t,N} (X_{[t] , [N]})  \Big\}_{t \in [T]} \right) \\
& \hspace{50pt} = \frac{H(U)}{NT} \\
&\hspace{50pt} = \frac{1}{NT} H(U_{1, [N]}, U_{2, [N]}, \dots, U_{T, [N]}) \\
& \hspace{50pt} \stackrel{{\rm (a)}}{=} \frac{1}{NT} \sum_{t \in [T]} H(U_{t, [N]} | U_{t-1, [N]}) \\
& \hspace{50pt} \stackrel{{\rm (b)}}{\le} \frac{1}{NT} \sum_{t \in [T]} \log(M_t + 1) \\
& \hspace{50pt} \stackrel{{\rm (c)}}{\le} \frac{1}{NT}  \sum_{t \in [T]} N( R_t  + \varepsilon) \\
& \hspace{50pt} = \frac{1}{T}  \sum_{t \in [T]}  ( R_t  + \varepsilon) \\
& \hspace{50pt} = \frac{1}{T}  \sum_{t \in [T]}  ( I_{\mu , \tpi} (X_{[t]}; U_t | U_{[t-1]})  + \varepsilon) \\
&  \hspace{50pt} \stackrel{{\rm (d)}}{=} \frac{1}{T} I_{\mu, \tpi}(X_{[T]} ; U_{[T]}) + \varepsilon
\end{split}
\end{equation*}
where (a) is by the chain rule for entropy; (b) uses the fact that conditioning reduces entropy and the fact that $U_{t,[N]}$ can take at most $M_t + 1$ values by construction of $\gamma_{t,N}$; (c) is by the choice of $M_t$'s; and (d) uses the chain rule for mutual information.
\end{proof}

\begin{lemma}
\label{lem:quant}
Let $\sX$ and $\sU$ be standard Borel spaces, and consider a pair $( \mu , \tpi ) \in \cP(\sX_{[T]}) \times \overrightarrow{\cM}(\sU_{[T]} | \sX_{[T]})$, such that $R_t = I_{\mu, \tpi}(X_{[t]} ; U_t | U_{[t-1]}) < \infty$ for all $t \in [T]$. Let $(X_{[T] , n} , U_{[T] , n})_{n \in [N]}$ be $N$ i.i.d.\ draws from the strategic measure $\Pr_\mu^{\tpi} $. Let $A_{t} \in \cB(\sX_{t,[N]}\times \sU_{t,[N]})$, $t \in [T]$, be a collection of Borel sets. Then, for a given $\varepsilon > 0$ and any sequence of positive integers $M_t$, there exist measurable mappings $g_t  : \sX_{[t], [N]} \rightarrow \sU_{t, [N]}$, $t \in [T]$, such that for each $t \in [T]$ and each $x_{[t-1] , [N]}$, $g_t(x_{[t-1],]N},\cdot)$ takes at most $M_t + 1$ values, and
\begin{align}
\label{eq:bound}
& \Pr_{\mu} \left[ ( X_{t, [N]} , g_t(X_{[t], [N]}) ) \in A_t \right]  \nonumber\\
 & \le \sum_{s \in [t]} \Big\{ \Pr_{\mu}^{\tpi} \left[  ( X_{s, [N]} , U_{s, [N]} ) \in A_s \right] \nonumber\\
 & ~~~~~~ + \Pr_{\mu}^{\tpi} \left[ i_s(X_{[s] , [N]} ; U_{[s], [N]} ) > N (R_s + \varepsilon) \right]  \nonumber\\
& ~~~~~~  +   \exp\left( -M_s e^{-N (R_s + \varepsilon) } \right) \Big\} ,
\end{align}
where
\begin{align*}
& i_s(x_{[s],[N]}; u_{[s],[N]}) \\
&\deq \log \frac{\mathrm{d} \Pr^{\mu, \tpi}_{U_{s , [N]} | (X_{[s],[N]}, U_{[s-1],[N]}) = ( x_{[s],[N]} , u_{[s-1],[N]}) }  }{\mathrm{d} \Pr^{\mu, \tpi}_{U_{s , [N]} | U_{[s-1],[N]}  = u_{[s-1],[N]} } }(u_{s,[N]})
\end{align*}
are the conditional information densities.
\end{lemma}

\begin{proof}


We will use a random sequential selection procedure to construct the sequence of mappings $g_1, g_2, \ldots, g_T$. The overall idea is a generalization of the proof of the achievability part of the lossy source coding theorem (see, e.g.,~\cite{gallager1968information} or \cite{gray2011entropy}).

Given $\Pr_\mu^{\tpi}$, let $\nu \in \cP(\sU_{[T]})$ denote the marginal distribution of the action process and disintegrate it as
\[
\nu(\mathrm{d} u_{[T]}) = \displaystyle\bigotimes_{t \in [T]} \nu_t(\d u_t | u_{[t-1]}).
\]
For each $t \in [T]$, define the Markov kernel $\nu_{t , [N]} \in \cM(\sU_{t, [N]} ~ | ~ \sU_{[t-1] , [N]})$ via
\[
\nu_{t , [N]} (\mathrm{d} u_{t , [N]} | u_{[t-1] , [N]}) \deq  \displaystyle\bigotimes_{n \in [N]} \nu_t(\d u_{t , n} | u_{[t-1] , n}). 
\]
In order to construct the finite-range mappings $g_t (\cdot): \sX_{[t],[N]} \rightarrow \sU_{[N]} $, we first choose the elements of $\sU_{[N]}$ to make up the range of $g_t$ and then specify how to assign one of these elements to each $x_{[t] , [N]}$ in the domain $\sX_{[t] , [N]}$.

For the first step, pick an arbitrary tuple $u_{[N]}(0) \in \sU_{[N]}$. Then let $u_{[N]}(1) , \dots, u_{[N]}(M_1)$ be i.i.d.\ draws from $\nu_{1 , [N]}$, and take the set $\{ u_{[N]}(0) , \dots, u_{[N]}(M_1)\}$ as the finite range of $g_1$. For the second step and for each $i_1 \in [M_1]$, let $u_{[N]}(i_1, 1) , \dots , u_{[N]}(i_1, M_2)$ be $M_2$ i.i.d. draws from $\nu_{2 , [N]}(\cdot | u_{[N]}(i_1)) $, and take the set
\begin{equation*}
\{ u_{[N]}(i_1 , i_2) \}_{i_1 \in [M_1] , i_2 \in [M_2]} \cup \{ u_{[N]}(0) \}
\end{equation*}
as the range of $g_2$. This process is continued inductively at each $t$: for each $(i_1, \dots, i_{t-1})$, we let 
\begin{equation*}
u_{[N]}(i_1 , \dots,  i_{t-1} , 1) ,  \dots , u_{[N]}(i_1 , \dots,  i_{t-1} , M_t)
\end{equation*}
be $M_t$ i.i.d.\ draws from $\nu_{t , [N]} (\cdot | u_{[N]}(i_1, \dots, i_{t-1}))$, and take
\begin{equation*}
  \{ u_{[N]}(i_1 , \dots,  i_{t-1} , i_t) \}_{i_1 \in [M_1] , \dots, i_t \in [M_t]} \cup \{u_{[N]}(0) \}
 \end{equation*}
as the range of $g_t$. Evidently, the range of each $g_t$ is selected at random conditionally on the realizations of the ranges of $g_1,\ldots,g_{t-1}$. The resulting collection of elements of $U_{[N]}$ can be arranged on a rooted tree of depth $T$, where the root has $M_1$ children, each depth-1 node has $M_2$ children, etc. Following Tatikonda \cite{TatikondaThesis}, we refer to this construction as a {\itshape tree code}.  
  
We now complete the construction of the $g_t$'s. To that end, we use the following recursive procedure. For $t \in [T]$ and $x_{[t] , [N]}$, suppose that $g_1, \ldots , g_{t-1}$ have already been specified.
Given $x_{[t] , [N]} $, if $g_{t-1}(x_{[t-1] , [N]})  = u_{[N]}(0)$, then we let $g_t(x_{[t] , [N]})  = u_{[N]}(0)$. Otherwise, if  $g_{t-1} (x_{[t-1] , [N]}) = u_{[N]}(i_1, \dots, i_{t-1}) $, consider the set
\begin{align*}
& G_t(x_{[t], [N]}) \\
&\deq \{ u_{[N]}(i_1, \dots, i_{t-1} , j)  : (x_{t, [N]} , u_{[N]}(i_1, \dots , i_t , j) ) \notin A_t \} .
\end{align*}
If it is empty, then let $g_t(x_{[t], [N]}) = u_{[N]}(0)$; otherwise, let 
\begin{equation*}
g_t(x_{[t], [N]}) =  u_{[N]}(i_1, \dots, i_{t-1} , j^*) 
\end{equation*}
where $j^* \in [M_t]$ is the smallest index $j$ of all $u_{[N]}(i_1, \dots, i_{t-1} , j) \in G_t(x_{[t] , [N]})$.

For each $t$, let $E_t$ denote the event that $g_t(X_{[t] , [N] }) = u_{[N]} (0)$. We upper-bound the probability of $E_t$, with respect to both $\mu$ and the random choice of the ranges of $g_1, \ldots, g_T$. By construction of $g_t$'s, $E_t$ will occur if either $E_{t-1}$ has occurred or if $G_t(x_{[t],[N]}) = \varnothing$ and none of $E_1, \ldots, E_{t-1}$ have occurred. Therefore, 
\begin{align}
\label{eq:recursiveE}
&\Pr[E_t] \le \Pr[E_{t-1}] \nonumber\\
& \qquad + \Pr [\{ G_t(X_{[t],[N]}) = \varnothing \} \cap E^c_1 \cap \dots \cap E^c_{t-1}].
\end{align}
By symmetry and independence in the generation of the ranges of $g_t$ {(in particular, using the fact that the range of $g_t$ is generated by drawing $M_t$ i.i.d.\ samples from $\nu_{t,[N]}(\cdot|u_{[t-1],[N]})$, plus the `error' tuple $u_{[N]}(0)$)}, we can estimate the second term on the right-hand side of \eqref{eq:recursiveE} by
\begin{align}
 \label{eq:division}
& \Pr [\{ G_t(X_{[t],[N]}) = \varnothing \} \cap E^c_1 \cap \dots \cap E^c_{t-1}]\nonumber \\
&  \le \int \Pr^{\tpi}_\mu(\mathrm{d}x_{[t], [N]}) \Pr^{\tpi}_\mu(\mathrm{d} u_{[t-1], [N]}) \nonumber\\
& \qquad  (1 - \nu_{t , [N]} (( x_{t , [N]} , U_{t , [N]} )  \in A_t^c  ~ | ~ u_{[t-1] , [N]}) )^{M_t} ,
\end{align}
{where we adhere to the standard convention of denoting random variables by uppercase letters and using lowercase letters for deterministic quantities. Thus, $\nu_{t , [N]} (( x_{t , [N]} , U_{t , [N]} )  \in A_t^c  ~ | ~ u_{[t-1] , [N]})$ is shorthand for
$$
\nu_{t , [N]} ( U_{t,[N]} : ( x_{t , [N]} , U_{t , [N]} )  \not\in A_t  ~ | ~ u_{[t-1] , [N]}).
$$}
Moreover, if we define the sets
\begin{align*}
& A_t (x_{t , [N]}) \deq  \left\{ u_{t, [N]} : (x_{t, [N]} , u_{t , [N]}) \in A_t \right\} , \\
& B_t (x_{[t] , [N]} , u_{[t-1] , [N]})   \\
& \quad \deq \left\{ u_{t, [N]} : i_t (x_{[t], [N]} , u_{[t] , [N]} ) > N (R_t + \varepsilon) \right\},
\end{align*}
then, performing a change of measure, we get
\begin{equation*}
\begin{split}
&\nu_{t , [N]} (( x_{t , [N]} , U_{t , [N]} )  \in A_t^c  ~ | ~ u_{[t-1] , [N]}) \\
&\ge \nu_{t , [N]} ( A_t^c (x_{t , [N]})  \cap B_t^c(x_{[t] , [N]} , u_{[t-1] , [N]})~ | ~ u_{[t-1] , [N]}) \\
& = \int \Pr_{U_{t , [N]} | X_{[t] , [N]} , U_{[t-1] , [N]}  } (\d u_{t , [N]} | x_{[t] , [N]} , u_{[t-1] , [N]}) ) \\
& \qquad \qquad \mathrm{1}_{\{ u_{t , [N]} \in  ( A_t^c (x_{t , [N]})  \cap B_t^c(x_{[t] , [N]} , u_{[t-1] , [N]})\}}  \\
& \qquad \qquad \cdot  \exp[- i_t (x_{[t] , [N]} ; u_{[t] , [N]})] \\
& \ge e^{-N ( R_t + \varepsilon)} \\
&~~ \int \Pr_{U_{t , [N]} | X_{[t] , [N]} , U_{[t-1] , [N]}  } (\mathrm{d} u_{t , [N]} | X_{[t] , [N]} , U_{[t-1] , [N]}) ) \\
& \quad  \quad  \quad  \quad \quad \quad \mathrm{1}_{\{ u_{t , [N]} \in  ( A_t^c (x_{t , [N]})  \cap B_t^c(x_{[t] , [N]} , u_{[t-1] , [N]})\}} .
\end{split}
\end{equation*}
Using the inequality $(1 - a b )^M \le 1 - b + e^{-M a}$ for $a , b \in [0,1]$ , we can estimate
\begin{equation}
\label{eq:finalestimate}
\begin{split}
& \left(1 - \nu_{t , [N]} (( x_{t , [N]} , U_{t , [N]} )  \in A_t  ~ | ~ u_{[t-1] , [N]}) \right)^{M_t} \\
& \le 1 - \Pr_\mu^{\tpi} [ (X_{t , [N]} , U_{t , [N]}) \in A_t^c ~ , \nonumber\\
& \qquad i_t (X_{[t], [N]} ; U_{[t] , [N]} ) \le N(R_t + \varepsilon) ~ | ~  x_{[t] , [N]} , u_{[t-1] , [N]} ] \\
& + \exp \left( -M_t e^{-N (R_t + \varepsilon) } \right).
\end{split}
\end{equation}
From Eqs.~\eqref{eq:recursiveE}--\eqref{eq:finalestimate}, it therefore follows that
\begin{equation*}
\begin{aligned}
\Pr[E_t] \le & ~\Pr[E_{t-1}] +  \Pr_\mu^{\tpi} \left[ (X_{t , [N]} , U_{t , [N]}) \in A_t \right] \\
+ & \Pr_\mu^{\tpi} \left[  i_t (X_{[t], [N]} ; U_{[t] , [N]} ) > N (R_t + \varepsilon)  \right] \\
+ & \exp \left( -M_t e^{-N (R_t + \varepsilon) } \right).
\end{aligned}
\end{equation*}
Solving this recursion, we obtain the bound
\begin{equation*}
\begin{aligned}
\Pr[E_t] \le & \sum_{t \in [T]} \big\{ \Pr_\mu^{\tpi} \left[ (X_{t , [N]} , U_{t , [N]}) \in A_t \right] \\
+ & \Pr_\mu^{\tpi} \left[  i_t (X_{[t], [N]} ; U_{[t] , [N]} ) > N(R_t + \varepsilon) \right] \\
+ & \exp \left( -M_t e^{-N (R_t + \varepsilon) } \right) \big\} .
\end{aligned}
\end{equation*}
for every $t \in [T]$. By construction,
\[
\Pr[(X_{t , [N]} , g_t(X_{[t]  , [N]}) ) \in A_t] \le \Pr[E_t], \qquad t \in [T] ,
\]
where the probability is w.r.t.\ the joint law of $X$ and the randon selections of $g_1 , \dots, g_T$. Therefore, there exists at least one choice of $g_1 , \dots, g_T$ satisfying \eqref{eq:bound}.
\end{proof}

\end{appendices}





\end{document}